\documentclass[11pt,letterpaper]{scrartcl}
\usepackage[top=1 in, bottom=1 in, left=1 in, right=1 in, letterpaper]{geometry}

\usepackage[english]{babel}
\usepackage[utf8]{inputenc}
\usepackage[T1]{fontenc}
\usepackage{lmodern}
\usepackage{etoolbox,mathtools}

\usepackage{algorithm}
\usepackage[noend]{algorithmic}
\usepackage{amsmath,amsthm,csquotes}
\usepackage{amssymb}
\usepackage{microtype}
\usepackage[dvipsnames]{xcolor}
\usepackage{tikz}
\usepackage{array}
\usepackage{enumerate}
\usepackage[shortlabels, inline]{enumitem}

\usepackage{esvect}

\usepackage{wasysym}

\usepackage{thmtools}
\usepackage{thm-restate}

\usepackage{hyperref}
\usetikzlibrary{calc,arrows,patterns,shapes,positioning,decorations.pathmorphing,
  trees, fit}
	
\usepackage{lmodern}
\usepackage{etoolbox,mathtools}
\usepackage{amsmath,amsthm,csquotes}
\usepackage{amssymb}
\usepackage{array}

\newif\ifdraft
\draftfalse

\ifdraft
\usepackage[draft]{movetoappendix}
\else
\usepackage{movetoappendix}
\fi

\setcapindent{0pt}
\usepackage{mparhack}
\newcounter{note}[section] \renewcommand{\thenote}{\thesection.\arabic{note}}
\newcommand{\hrnote}[1]{\refstepcounter{note}\textcolor{Blue}{%
    \mathversion{bold}\marginpar{\hfill\tiny\sffamily\bfseries
      \textcolor{Blue}{HR~\thenote}}$\ll$\bfseries\sffamily#1
    --Harry$\gg$}\mathversion{normal}}

\ifdraft
\else
\def\menote #1{}%
\def\hrnote #1{}%
\def\rsnote #1{}%
\fi

\newtheorem{theorem}{Theorem}
\newtheorem{observation}[theorem]{Observation}
\newtheorem{lemma}[theorem]{Lemma}
\newtheorem{claim}[theorem]{Claim}

\title{Compact Oblivious Routing}

\author{Harald R{\"a}cke$^1$ \quad 
Stefan Schmid$^2$\\
{\footnotesize $^1$ TU Munich, Germany}\quad 
{\footnotesize $^2$ University of Vienna, Austria}
}

\def\capac{\operatorname{cap}}
\def\out{\operatorname{out}}

\date{}

\begin{document}

\sloppy

\maketitle

\begin{abstract}
Oblivious routing is an attractive 
paradigm for large distributed systems in which 
centralized control and frequent reconfigurations 
are infeasible or undesired (e.g., costly).
Over the last almost 20 years, much progress
has been made in devising oblivious routing schemes
that guarantee close to optimal load and also algorithms
for constructing such schemes efficiently have been designed.
However, a common drawback of
existing oblivious routing schemes
is that they are not compact:
they require large routing
tables (of polynomial size), which does not scale.

This paper presents the first oblivious routing scheme which guarantees
close to optimal load and is compact at the same time --
requiring routing tables of \emph{polylogarithmic} size. Our algorithm
maintains the polynomial runtime and polylogarithmic competitive ratio of
existing algorithms, and is hence particularly well-suited for emerging
large-scale networks.
\end{abstract}

\thispagestyle{empty}
\addtocounter{page}{-1}

\section{Introduction}

\subsection{Motivation}

With the increasing scale and dynamics
of large networked systems, observing and reacting to
changes in the workload and reconfiguring the 
routing accordingly 
becomes more and
more difficult. Not only does  a larger network and more 
dynamic workload require more fine-grained monitoring
and control (which both introduce overheads), also the process of
re-routing traffic itself (see e.g.~\cite{ieeesurvey}) can lead to temporary performance degradation
and transient inconsistencies.

Oblivious routing provides an attractive 
alternative which avoids these reconfiguration overheads 
while being \emph{competitive}, i.e.,
while guaranteeing route allocations which are
almost as good as
adaptive solutions. 
It is hence not surprising that
oblivious routing has received much attention 
over the last two decades. Indeed, today,
we have a good understanding of fast (i.e., polynomial-time)
and ``competitive'' oblivious routing algorithms (achieving a polylogarithmic
approximation of the load, which is optimal). 

However, while oblivious routing 
seems to be the perfect paradigm  
for emerging large networked systems,
there is a fly in the ointment. Oblivious routing algorithms
require large routing tables: namely \emph{polynomial} 
in the network size. This is problematic, as fast 
memory in routers is expensive, not only in terms of monetary
costs but also in terms of power consumption.

The goal of this paper is to design oblivious routing schemes which only require
small routing tables (which are \emph{compact}), and that at the same time
still guarantee a close-to-optimal load.

\subsection{The Problem in a Nutshell}
The network is given as an 
undirected graph $G=(V,E)$ with $n$ vertices.
The edges $E$ are weighted by a capacity function
$\capac: V\times V\rightarrow \mathbb{R}^+_0$;
if $\{x,y\}\in E$, the function returns $0$,
otherwise a positive value.

The \emph{oblivious routing problem} is to set up
a unit flow for each source-target pair $(s, t) \in V \times V$ 
that determines how demand between $s$ and $t$ is routed in the
network $G$. This unit flow is pre-specified without knowing
the actual demands. When a demand vector $\vec{d}$ is given that
specifies for each pair of vertices the amount of traffic to be
sent, the demand-vector is routed by simply scaling the unit
flow between a pair $(s, t)$ by the corresponding demand $d_{st}$
between the two vertices. 
This means that traffic is
routed along \emph{pre-computed paths} and that no path-selection
is done dynamically.

The congestion $C_{\text{obl}}(G, \vec{d})$ 
resulting from a given oblivious
routing scheme, is then compared to the optimal possible congestion
$C_{\text{opt}}(G, \vec{d})$ that can be obtained for demand vector 
$\vec{d}$ 
in $G$.
The \emph{competitive ratio} of the oblivious routing scheme is
defined as 
$$
\max_{\vec{d}} \frac{C_{\text{obl}(G, \smash{\vec{d}})}}{C_{\text{opt}(G, \smash{\vec{d}})}}
$$

In this paper, we are mainly interested in \emph{compact}
solutions which minimize the number of required
forwarding rules. 
We say that an oblivious routing scheme is
\emph{compact} if the number of rules
required for a vertex $v$ is
in $O(\text{polylog}(n)\cdot \Delta)$,
where $n$ is the number of vertices in the network
and $\Delta$ the maximal vertex degree of $G$.
In other words, a polylogarithmic number
of rules are required \emph{per link}.
In addition to the competitive ratio, the runtime,
and the table size, we are also interested in the 
required vertex labels (i.e., their size) and the 
required packet header size.


\subsection{Our Contributions}

This paper presents the first \emph{compact} oblivious
routing scheme. 
Our approach builds upon 
an oblivious path selection scheme based on
classic decomposition trees, which is then adapted
to improve scalability, and in particular, to ensure
small routing tables and message headers, while
preserving polynomial runtime and a polylogarithmic
competitive ratio.

We present two different implementations of our approach
and our results come in two different flavors accordingly
(more detailed theorems
will follow): 

\begin{theorem}
There exist polynomial-time algorithms
which achieve a polylogarithmic competitive ratio w.r.t.\ the congestion
and
require routing tables of polylogarithmic size for 
\begin{enumerate}
\item networks with arbitrary edge capacities 
which have a decomposition tree of bounded degree, and for
\item arbitrary networks with uniform edge capacities.
\end{enumerate}
Our algorithms only require small
(polylogarithmic) header sizes and vertex labels.
\end{theorem}

Networks for which there are decomposition trees of small degree include for
example (constant-degree) grids or graphs that exclude a minor of constant
size. The exact requirements that a decomposition tree has to fulfill will be
given later.

\section{Algorithm and Analysis}\label{sec:algo-analysis}

This section describes an oblivious path selection scheme for
general undirected networks that obtains close to optimal congestion 
and can be implemented with routing tables and routing headers of small size.
In a nutshell, our algorithm leverages a path selection scheme 
for general networks that
guarantees a good competitive ratio w.r.t.\ congestion, 
and then adapts it so that it
can be implemented with small space requirements. 
We discuss the two phases of this algorithm in turn.

\subsection{Oblivious Path Selection Scheme}\label{sec:description}

There exist essentially two
path selection schemes that could be used as a basis for our
approach. First, there is the original result by Räcke~\cite{Rae02}
who showed that oblivious routing with a polylogarithmic competitive ratio is
possible in general networks, using a \emph{hierarchical path
  selection scheme} (cf.~Section~\ref{sec:description}) that guarantees a
competitive ratio of $O(\log^3n)$.  
Second, we could use the result
by Räcke et al.~\cite{racke2014computing} that computes a hierarchical decomposition in nearly
linear time, which also gives a hierarchical path selection
scheme albeit with weaker parameters than~\cite{Rae02}.
The result in~\cite{Rae02} has been improved
to a competitive ratio of $O(\log n)$ with a different scheme in~\cite{racke2008optimal}.
The latter scheme can be roughly viewed as a convex combination of spanning
trees\footnote{This is not entirely correct as the trees are not proper
  spanning trees but the difference is not important for the above discussion.}
A path between a vertex $s$ and a vertex $t$ is chosen by sampling a random spanning
tree and then choosing the path between $s$ and $t$ in this tree. 

In this paper, we will build upon the original result~\cite{Rae02}
which we call the \emph{hierarchical
  path selection scheme}.
The challenge with implementing the path selection mechanism in~\cite{racke2014computing,racke2008optimal} space-efficiently is that the
number of spanning trees is quite large (polynomial in $n$).
It seems difficult
to avoid that a vertex in the graph has to store some information for every tree, which
yields routing tables of polynomial size.

The hierarchical
  path selection scheme is based on a hierarchical decomposition of the graph
$G=(V,E)$. The vertex set $V$ is recursively partitioned into smaller and
smaller pieces until all pieces contain just single vertices of $G$. We will
refer to the pieces/subsets arising during this partitioning process as \emph{clusters}.

To such a recursive partitioning corresponds a decomposition tree
$T=(V_T,E_T)$. A vertex $x$ in this tree corresponds to cluster
$V_x\subseteq V$ and there is an edge between a parent node $p$ and a child
node $c$ if the cluster $V_c$ arises from partitioning $V_p$. The root $r$ of
$T$ corresponds to the subset $V_r=V$ and the leaf vertices correspond to
singleton sets $\{v\}, v \in V$.

In order to simplify the notation and description we assume that all leaf
vertices in $T$ have the same distance to the root (this could e.g., be achieved
by introducing dummy partitioning steps in which a set is partitioned into
itself).
We use $h$ to denote the height of the tree.  
Let for a vertex $v\in V$, $a_\ell(v)$ denote the
ancestor of $\{v\}$ on level $\ell$ of the tree, 
where the level of a vertex is
its distance from the root. Here we use $\{v\}$ 
as a shorthand for \enquote{the
    leaf node that corresponds to cluster $\{v\}$}.
The \emph{$\ell$-weight of $v$} is the weight of all edges incident 
to $v$ that 
leave the cluster $V_{a_\ell(v)}$. Formally
$w_\ell(v):=\sum_{e=\{v,x\}:x\notin V_{a_\ell(v)}}\capac(e)$. 
We extend this
definition to subsets of $V$ by setting $w_\ell(U):=\sum_{u\in U}w_\ell(u)$ 
for every
subset $U\subseteq V$.

We also introduce for every cluster $S$ in the decomposition tree
a weight function $w_S: S\mapsto\mathbb{R}_0^+$ and a weight function
$\out_S: S\mapsto\mathbb{R}_0^+$. For a level $\ell$-cluster $S$ we define
$w_S:=w_{\ell+1}\restriction_S$ and $\out_S:=w_{\ell}\restriction_S$.
Note that $\out_S$ counts edges that connect vertices of $S$ to vertices
outside of $S$ while $w_S$ also counts edges that connect different
sub-clusters of $S$. We refer to $w_S$ as the \emph{cluster-weight} of $S$ and
to $\out_S$ as the \emph{border-weight} of $S$.

Using this weight definition, we define a \emph{concurrent multicommodity flow problem}
(CMCF-problem) for every cluster $S$ in the decomposition tree.  
For every (ordered) pair $(u,v)$ there is a demand of 
$w_{S}(u)w_{S}(v)/w_S(S)$. Informally speaking, this means that
every vertex injects a total flow that is equal to its $w_S$-weight and
distributes this flow to the other vertices in $S$, proportionally to the
$w_S$-weight of these vertices. We will use the decomposition
tree $T$ 
in~\cite{Rae02} with the following properties:
\begin{itemize}
\item the height of $T$ is $O(\log n)$, and
\item for every cluster $S$ in the decomposition tree, 
the CMCF-problem for $S$ can be solved with congestion at most
$C=O(\log^2n)$ \emph{inside} $S$.
\end{itemize}

Now suppose that we are given a decomposition tree with these properties. The path
selection in~\cite{Rae02} is then performed as follows. Suppose that 
we want to
choose a path between vertices $s$ and $t$ in $G$. Let $x_s$ and $x_t$ denote
the leaf vertices in $T$ that correspond to singleton clusters $\{s\}$ and
$\{t\}$, respectively. Let $x_s=x_1, x_2,\dots,x_k=x_t$ denote the vertices in
the tree on the path from $x_s$ to $x_t$. We first choose a random vertex $v_i$ from
each cluster $V_{x_i}$ according to the cluster-weight, i.e., the probability
that $v$ is chosen is $w_{V_{x_i}}(v)/w_{V_{x_i}}(V_{x_i})$. Note that $v_1=s$
and $v_k=t$ as the corresponding clusters just contain a single vertex. It
remains to select a path that connects the chosen vertices.

Suppose we want to connect two consecutive vertices $v_p$ and $v_c$, where
$V_{x_p}$ is the parent cluster of $V_{x_c}$.
We choose an intermediate vertex $\alpha$ inside $V_{x_c}$ according to the
border-weight of $V_{x_c}$, i.e., the probability that $v$ is chosen is
$\out_{V_{x_c}}(v)/\out_{V_{x_c}}(V_{x_c})$. We then consider the
solution to the CMCF-flow problems for $V_{x_c}$ and $V_{x_p}$. The first
solution contains a flow $f(c,\alpha)$ between $v_{x_c}$ and $\alpha$, and the second
contains a flow $f(p,\alpha)$ between $v_{x_p}$ and $\alpha$. We sample a
random path from each flow. Concatenating these two paths, gives a flow between
$v_{c}$ and $v_{p}$. 
For the following analysis we call the sub-path between
$x_c$ and $\alpha$ the \emph{lower sub-path} and the path between $\alpha$
and $x_p$ the \emph{upper sub-path}.

Concatenating all vertices $v_i$ in the above manner gives a path between $x_s$
and $x_t$. In the following we analyze the expected load generated on an edge
due to this path selection scheme under the condition that an optimal algorithm
can route the demand with congestion $C_{\mathrm{opt}}$.
For completeness and as we will need to modify this proof later, 
we repeat the following observations from~\cite{Rae02}.

\begin{lemma}\label{lem:edgeload}
The expected load on an edge is at most $O(h\cdot C\cdot C_{\mathrm{opt}})$.
\end{lemma}
\begin{proof}
Fix an edge $e$ for which both end-points are contained in some cluster $S$.
Let $S_{1}$,\dots,$S_{r}$ denote the child-clusters of $S$. We first
analyze the total demand that we have to route between a pair of vertices
$(a,b)\in S\times S$ due to an upper sub-path where $a$ is chosen as the
intermediate vertex $\alpha$
and $b$ is chosen as a random vertex from the parent cluster $S$. 
Assume $a\in S_{i}$ for some
child cluster $S_i$. Then the probability that we choose $a$ as $\alpha$ 
is
$\Pr[\text{$a$ is chosen}]=\out_{S_{i}}(a)/\out_{S_{i}}(S_{i})$. The probability that we choose $b$ as the random
end-point in $S$ is $\Pr[\text{$b$ is chosen}]=w_{S}(b)/w_{S}(S)$. Note that any message
for which we route between the child cluster $S_{i}$ and the parent cluster $S$
has to leave or enter the cluster $S_i$. Therefore the total demand for
these messages can be at most $C_{\mathrm{opt}}\cdot\out_{S_i}(S_i)$, as otw.\ an optimum
congestion of $C_{\mathrm{opt}}$ would not be possible. Hence, the expected
demand for pair $a$ and $b$ is only
\begin{equation}
\label{eqn:demandOne}
\begin{split}
\out_{S_{i}}(S_{i})C_{\mathrm{opt}}\cdot
\Pr[\text{$a$ is chosen}]\cdot
\Pr[\text{$b$ is chosen}]
&=
\out_{S_{i}}(S_{i})C_{\mathrm{opt}}
\cdot\frac{\out_{S_{i}}(a)}{\out_{S_{i}}(S_{i})}
\cdot
\frac{w_{S}(b)}{w_{S}(S)}\\
&=
\frac{w_{S}(a)\cdot w_{S}(b)}{w_{S}(S)}\cdot C_{\mathrm{opt}}\enspace,
\end{split}
\end{equation}
where we used the fact that $\out_{S_{i}}(a)=w_{S}(a)$, which holds since
$S_{i}$ is a direct child-cluster of $S$.

Now we analyze the demand that is induced for a pair $(a,b)\in S\times S$ 
due to the lower
part of a message between $S$ and its parent cluster. We assume that $a$
is chosen as the intermediate vertex $\alpha$ and $b$ is chosen as a random
node in the child-cluster $S$. The probability that $a$ is chosen as
intermediate vertex is $\Pr[\text{$a$ is chosen}]=\out_S(a)/\out_S(S)$ and the
probability that $b$ is chosen is $\Pr[\text{$b$ is chosen}]=w_S(b)/w_S(S)$.
Every such message has either to leave or enter cluster $S$. Hence, the total
demand for these messages induced on pair $(a,b)$ is at most
\begin{equation}
\label{eqn:demandTwo}
\begin{split}
\out_{S}(S)C_{\mathrm{opt}}\cdot
\Pr[\text{$a$ is chosen}]\cdot
\Pr[\text{$b$ is chosen}]
&=
\out_{S}(S)C_{\mathrm{opt}}\cdot
\frac{\out_{S}(a)}{\out_{S}(S)}
\cdot
\frac{w_{S}(b)}{w_{S}(S)}\\
&\le
\frac{w_{S}(a)\cdot w_{S}(b)}{w_{S}(S)}\cdot C_{\mathrm{opt}}\enspace,
\end{split}
\end{equation}
where we used the fact that $\out_S(a)\le w_S(a)$.

Combining Equation~\ref{eqn:demandOne} and Equation~\ref{eqn:demandTwo} gives
that the messages involving cluster $S$ induce a demand of only
$2{w_{S}(a)\cdot w_{S}(b)}/{w_{S}(S)}\cdot C_{\mathrm{opt}}$ between vertices $a$ and
$b$ from $S$. Since we route this demand according to the multicommodity flow
solution of the CMCF-problem for cluster $S$, the resulting load is at most
$2C\cdot C_{\mathrm{opt}}$ on any edge \emph{inside} cluster $S$, while edges
not in $S$ have a load of zero. Summing the load induced by messages for all
clusters and exploiting the fact that an edge is at most contained in $h$
different clusters, gives a maximum load of $2hC\cdot C_{\mathrm{opt}}$, i.e.,
a competitive ratio of $2hC$.
\end{proof}

\subsection{Implementation A: Decomposition Trees with Small Degree}

We now present a space efficient implementation 
of the above path selection scheme. In the following, we will
assume that the maximum degree of
the decomposition tree $T$ is small.

The basic building block for our implementation is a method 
that given a random starting point $v\in S$ chosen according to the 
cluster-weight of $S$ (i.e.,
the probability of choosing $v$ is $w_S(v)/w_S(S)$), routes to a random node
$v_i\in S_i$ chosen according to the border weight of $S_i$. Here $S_i$ is
either a child-cluster of $S$ (in case we want to communicate downwards in the
tree) or $S_i=S$ (in case we want to communicate upwards). In the following we
use $S_i, i\in\{1,\dots,r\}$ to denote the child-clusters of $S$ and $S_0=S$ to
denote $S$ itself. Let $G[S]$ denote the sub-graph induced by vertices in $S$.

For every $i\in \{0,\dots,r\}$ we compute a single commodity flow $f_i$ in
$G[S]$ as follows. We add a super-source $s$ and connect it to every vertex
$v\in S$ with an edge of capacity $w_S(v)\cdot\out_{S_i}(S_i)$ and a
super-target $t$ to which every vertex in $v\in S$ connects with capacity
$\out_{S_i}(v)\cdot w_S(S)$. Note that all source edges together have the same
capacity as the target edges.

We compute a maximum $s$-$t$ flow in $G[S]$ after scaling the
capacities of edges in $G[S]$ up by $w_S(S)\cdot C$. This flow will saturate edges
from $s$ and to $t$ as these form the bottleneck in the network and, hence, it
will have a value of $w_S(S)\cdot\out_{S_i}(S_i)$. This follows
from the fact that in $G[S]$ (without scaling) every node can inject
a flow of $\out_S(v)w_S(S)$ and distribute it so that a node $v'$ receives
$\out_S(v)w_S(v')$ of this flow and the congestion is only $w_S(S)\cdot C$.

We store the flow $f_i$ in a distributed manner at the vertices of $S$,
as follows. Fix $v\in S$. For every edge we store how much flow enters 
or leaves $v$. In order
to route from the cluster-distribution of $S$ to the border-distribution for
$S_i$, $i\in\{0,\dots, r\}$, we choose random outgoing links (where a link is
taken with probability proportional to the outgoiong flow) until the chosen
link is the super-target~$t$. When we want to route from the
border-distribution of $s_i$ to the cluster-distribution of $S$, we take random
incoming links (where a link is chosen with probability proportional to the
incoming flow), until the chosen link corresponds to the super-source $s$.
The proof of the following claim is analogous to Lemma~\ref{lem:edgeload}.

\begin{claim}
The expected load of an edge due to the path selection scheme is only
$O(h\cdot \mathrm{deg}(T)\cdot C\cdot C_{\mathrm{opt}})$.
\end{claim}
\begin{proof}
Suppose that the optimum congestion is $C_{\mathrm{opt}}$. The total traffic
that the scheme has to route between the cluster-distribution of $S$ and the
border-distribution of $S_i$ is only
$\out_{S_i}(S_i)\cdot C_{\mathrm{opt}}$. We route this traffic according to
flow $f_i$ of value $\out_{S_i}(S_i)w_S(S)$. Hence, the maximum load of
an edge in $G[S]$ (according to original capacity) is $C\cdot C_{\mathrm{opt}}$.

Since a cluster $S$ may have $\mathrm{deg}(T)$ many flows $f_i$ and an edge is
contained in $h$ different clusters the claim follows.
\end{proof}

\begin{claim}
The path selection scheme can be implemented with routing tables of size
$O(\deg(v)\deg(T)(\log m+\log W))$, labels of length
$O(h\log(\mathrm{deg}(T)))$,
and header length $O(h\log(\mathrm{deg}(T)))$.
\end{claim}
\begin{proof}
Since the capacities in the flow problem for $f_i$ are all integral, the flow
solution will be integral~\cite{ford1962flows}. Suppose 
that the original capacities of the graph
are integers in the range $\{1,\dots,W\}$. After scaling the capacity of graph
edges in $G[S]$, these have a capacity of at most $w_S(S)\cdot W \cdot C$
(note that we assume that $C$ is integral).  Edges from $s$ and to $t$ have
a capacity of $w_S(v)\out_{S_i}(S_i)$ and  $w_S(S)\out_{S_i}(v)$, respectively.
Using the fact that $w_S(S)$ and $\out_{S_i}(S_i)$ are at most $mW$, and
$d,C\le m$ we get that a number describing the flow value along an edge can be encoded with
\begin{equation*}
\log_2(m^2W^2)=O(\log(m)+\log(W))
\end{equation*}
many bits. Hence, a node $v$ has to store only $O(\deg(v)\deg(T)(\log m+\log W))$
many bits.

For the routing scheme we relabel the vertices. We number the children of a
vertex in the tree and encode a leaf vertex by its path from the root. This
generates labels of $O(h\log(\mathrm{deg}(T)))$ bits. The routing algorithm now
only needs to have the label of the source vertex and the target vertex and a
marker that marks where in the tree the routing currently is.
\end{proof}

In summary, and leveraging the decomposition tree,
we have derived the following result:

\begin{theorem}[Decomposition Trees of Small Degree]
\label{thm1:bounded-depth}
For decomposition trees of constant degree,
there exists a polynomial-time algorithm
which achieves a competitive ratio of $O(\log^3{n})$,
requires routing tables of size
$O(\deg(v)\deg(T)(\log m+\log W))$, labels of length
$O(h\log(\mathrm{deg}(T)))$,
and header sizes of $O(h\log(\mathrm{deg}(T)))$.
\end{theorem}

\subsection{Implementation B: Uniform Capacities}

In this section we present a different implementation of the hierarchical
routing scheme, for scenarios where the decomposition trees
can be of arbitrary degree but where network capacities 
are uniform. Again the basic building block is to route from a node
chosen according to the cluster-distribution of some cluster $S$ to the border
distribution of $S_i$ where either $S_i=S$ or $S_i$ is a child-cluster of $S$.

\def\rS_#1{\|S_{#1}\|}

Assume that every edge in the graph $G$ has capacity $1$.
We round the outgoing capacity $\out_{S_i}(S_i)$ of a child-cluster 
$S_i$, $i\ge 0$ to the
next larger power of $2$ and denote the rounded value with $\rS_i$. We also
re-order the children w.r.t.\ this value, i.e., $S_1$ is the child-cluster
with smallest $\rS_i$-value. Since there are at most $m$ possible values for
$\out_{S_i}(S_i)$, there are only $\log m$ possible values for $\rS_i$. 
There are only
$\binom{r+\log m}{\log m}$ possibilities to choose the $\rS_i$-values of the
$r$ children of cluster $S$. Hence, we can store these with
$O(\log(m)\cdot\log(r))$ many bits. In addition we store the value of
$\rS_0$, which requires $O(\log\log m)$ bits, and the value of $w_S(S)$
which requires $O(\log m)$ bits.

In order to design the routing scheme for an individual cluster, we embed a
hypercube of dimension $d:=\lceil\log_2(\sum_{i\ge0}\rS_i)\rceil$. We first
order the hypercube nodes in an arbitrary way and then reserve a \emph{($i$-th)  range}
of $\rS_i$ consecutive hypercube nodes for every $i\ge 0$. Note that we store
the (rounded) size of all children and that it is straightforward to compute
the ranges assigned to any $i$ from this information.

Then we map the hypercube nodes to nodes of $S$. First we map hypercube nodes
in the $i$-th range to nodes with $\out_{S_i}(v)>0$ such that each node
receives at least $\out_{S_i}(v)$ and at most $2\out_{S_i}(v)$ hypercube nodes.
Hypercube nodes that remain unmapped after this step (i.e., nodes
that do not fall within any range) are mapped arbitrarily subject to the
constraint that a cluster node $v$ does not receive more than $4w_S(v)$
hypercube nodes. This can easily be done as the number of hypercube nodes
($2^d$) is at most $4\sum_i\sum_{v\in S_i}\out_{S_i}(v)=4(w_S(S)+\out_S(S))\le 8w_S(S)$.

\def\deg{\operatorname{deg}}

\begin{observation}
There are at most $8w_S(v)$ hypercube nodes mapped to any graph node.
\end{observation}
For the embedding we set up a concurrent multicommodity flow problem as
follows. For every edge $\{x,y\}$ of the hypercube that is mapped to endpoints
$\{v_x,v_y\}$, we introduce a demand of $1$ between $v_x$ and $v_y$ in both
directions. Then every node sends and receives a total traffic of at most
$8d\cdot w_S(v)$. By adding fake traffic we can turn this instance into
a balanced multicommodity flow instance in which every vertex sends and
receives a traffic of exactly $8d\cdot w_S(v)$.

We can solve this multicommodity flow instance with congestion at most
$16dC$ inside the cluster $S$ by using Valiant's trick~\cite{VB81,kolman2002improved} of sending to
random intermediate destinations and using the fact that each flow
can send a traffic of $w_S(v)$ to random destinations with congestion
$C$.

\subsubsection{Using the Hypercube}

How do we exploit the embedded hypercube? If during the routing scheme we are
required to send a message from a cluster node $v_p$ to a cluster node
$v_c\in S_i$ we proceed as follows. 
Instead of choosing an intermediate node
$\alpha$ according to probability distribution $\out_{S_i}(v)/\out_{S_i}(S_i)$
we choose a random hypercube node from the
$i$-th range. Then we route a message inside the hypercube to this node. For
this we let the message start at a random hypercube node from the nodes that are
mapped to $v_p$.

Note that this means that the probability that the message is sent to node $\alpha$ lies
between $\out_{S_i}(\alpha)/\rS_i$ and $2\out_{S_i}(v)/\rS_i$ as the hypercube nodes
in the $i$-th range are not mapped completely uniformly.

For the second part of the message we proceed analogously in the hypercube of
$S_i$. We let the message start at a random hypercube node mapped to $\alpha$
and choose a random hypercube node as its target.

Again due to the non-uniform mapping, the target distribution on $S_i$
(i.e., $w_{S_i}(v)/w_{S_i}(S_i)$) is not reached exactly, but deviations by
a constant factor might occur. This only influences the congestion of a single
step by a constant factor, but it could be problematic if we used this approach
along a path in the tree: in each step the distribution would change by a
constant factor.

Therefore, we add an additional step that fixes the distribution over $S_i$. We
embed an additional hypercube $H_S$ for every cluster $S$ with dimension
$\lceil \log_2(w_S(S))\rceil$. The mapping is done such that each
cluster-vertex $v\in S$ receives exactly $w_S(v)$ hypercube nodes among the
first $w_S(S)$ nodes from $H_S$ (the remaining nodes are distributed
uniformly). Since every node in the cluster $S$ stores the
value of $w_S(S)$, we can route from a node $v\in S$ to a random node chosen
according to distribution $w_{S}(v)/w_{S}(S)$, by just selecting a random
hypercube node from the first $w_S(S)$ nodes.

\subsubsection{Analysis}

We showed that the congestion for sub-messages that
involve cluster $S$ is small.
There are two types of such messages:
\begin{enumerate}
\item messages that start at an intermediate node (distributed according to the
border weight of $S_i$ for some $i\ge 0$) and are sent to a random node $v\in
S$ distributed according to the cluster-weight of $S$; and
\item messages that start at a  random node $v \in S$ and are sent to some
intermediate node. 
\end{enumerate}
It was shown that the total traffic that is sent between a pair $v_i$
and $v$, where $v$ is distributed according to the cluster weight of $S$ and
$v_i$ is distributed according to the border weight of $S_i$, is only
$\out_{S_i}(v_i)w_S(v)/w_S(S)\cdot C_{\mathrm{opt}}$. 

In our new scheme this changes slightly. For messages of the second type the
source is distributed as before but the target may have a slightly different
distribution (as we choose a random hypercube node in the $i$-th range). For
messages of the first type already the source may have a slightly different
distribution (as we choose a random hypercube node from some range in the
hypercube for a child- or parent-cluster). Also the target distribution is
slightly skewed as we choose a random hypercube node as the target. 

But since the distributions are only changed by a constant factor this
difference does not really influence our analysis. We still have the property
that the traffic between $v_i$ and $v_S$ is
$\Theta(\out_{S_i}(v_i)w_S(v)/w_S(S)\cdot C_{\mathrm{opt}})$.

The second difference is that the traffic is not sent according to the
CMCF-problem for cluster $S$ but it is instead sent along the hypercube. Note
that due to the embedding of the hypercube, a cluster node $v\in S_i$ has
$\Theta(\out_{S_i}(v))=\Theta(w_S(v))$ hypercube nodes in the $i$-th range
mapped to it (i.e., hypercube nodes are balanced perfectly up to constant factors). 
Hence the demand between $v_i$ and $v$ will be split among
$\Theta(\out_{S_i}(v_i)w_S(v))$ pairs in the cube. Therefore the demand
for every pair in the cube is only
$\Theta(C_{\mathrm{opt}}/w_S(S))=\Theta(C_{\mathrm{opt}}/2^d)$. This means
that at most a traffic of $O(C_{\mathrm{opt}})$ starts and ends at every vertex
and routing this traffic using Valiant's trick gives a congestion of
$O(d C_{\mathrm{opt}})$ in the hypercube. Since we embedded the hypercube with
congestion $O(dC)$, the congestion of a graph edge will be $O(d^2C\cdot
C_{\mathrm{opt}})$ (as each hypercube node has degree $d$), which gives rise to the following lemma.
\begin{lemma}
Implementation~B guarantees a maximum expected load of $O(hd^2C\cdot C_{\mathrm{opt}})$.
\end{lemma}
\begin{proof}
The lemma follows by applying the previous argument for each level of the tree.

It remains to bound the edge-load induced by the re-randomization steps. The
total traffic that is send to a cluster $S$ in the tree is at most
$(\sum_i\out(S_i))\cdot C_{\mathrm{opt}}=\Theta(w_S(S)\cdot C_{\mathrm{opt}})$.
For each such message we require a re-randomization, because in our current
scheme, it is only distributed approximately according to the cluster-weight of
$S$.

However by design each vertex receives exactly a $w_S(v)/w_S(S)$-fraction of
the re-randomization messages, and a $\Theta(w_S(v)/w_S(S))$-fraction of messages start at
$v$, since the messages are approximately distributed according to
cluster-weight. Sending these messages along the hypercube introduces
congestion $\Theta(d\cdot C_{\mathrm{opt}})$ in the cube and $\Theta(d^2C\cdot
C_{\mathrm{opt}})$ due to the embedding.
\end{proof}
\begin{lemma}
Implementation~B requires space $O(hC\log(m)\log\log(m)\deg(v))$ bits at every
vertex and a label and header length of $O(h\log(\deg(T)))$.
\end{lemma}
\begin{proof}
A vertex $v\in S$ has to store the approximate size $\rS_i$ of the
child-clusters of $S$. Summing this over all levels gives
$O(h\log(m)\cdot\log(r))$ bits. In addition one has to encode the embedding of
the hypercubes. The congestion of the solution to the concurrent multicommodity
flow problem for embedding a hypercube is $O(dC)$. This fractional solution
will encode a flow for every hypercube edge. Using a standard randomized
rounding approach, we can
route the flows to paths with a congestion of
$O(dC+\log(m))=O(dC)$.
This is done as follows. For every pair $\{x, y\}$ we take the
unit flow and first decompose this unit flow into flow-paths. Then we choose
for every pair one of the flow-paths at random (proportional to its weight).
Let $X_{i}(e)$ denote the random variable that describes whether the flow path
for the $i$-th pair includes edge $e$. By design the above process guarantees
$E[X_i(e)]=f_i(e)$, where $f_i(e)$ is the flow for pair $i$ that goes through
edge $e$. The total load on edge $e$ is $\sum_iX_i(e)$. This is a sum of
negatively correlated random variables with expectation $\mu=O(dC)$. Using
Lemma~\ref{lem:chernoff} (in the appendix) with $\delta=3\ln(m)/\mu$ gives 
that with constant probability, no edge exceeds load
$O(dC+\log m)$.

Therefore only $O(\deg(v)dC)$ paths traverse a vertex $v$ (recall that
$C\ge\log m$). For every path, we
need to store the outgoing edge and the id of the paths on this edge. This
requires $(\log_2(\deg(v))+\log_2(dC))$ bits for every path and
$O(d\deg(v)C\log(d\deg(v)C))$ bits in total. Multiplying with the height and
using $d=\Theta(\log m)$ gives
$O(h\log(m)\cdot(\log(\deg(T)+\deg(v)C\log\log(m)))
=O(hC\log(m)\log\log(m)\deg(v))$ bits.

The header- and label-length is analogous to Implementation~A. We just use the
root-to-leaf path in the tree as a label and a header consists of the
source-label, the target-label, and a marker.
\end{proof}

In summary we derived the following result:

\begin{theorem}[Compact Oblivious Routing for Uniform Capacities]
\label{thm2:uniform-cap}
For arbitrary networks of uniform capacities,
there exists a polynomial-time oblivious routing algorithm
which is $O(\log^3 n)$-competitive, 
and which requires $O(hC\log(m)\log\log(m)\deg(v))$ bits at every vertex and a
label and header length of $O(h\log(\deg(T)))$.
\end{theorem}

\section{Related Work}\label{sec:relwork}

The drawbacks of adaptive routing
have been discussed intensively in the literature,
see e.g.,~\cite{harry-survey} for a survey.
In particular, 
adaptive routing schemes need
global information about the routing problem
in order to calculate the best paths,
and even if it were possible to collect such information 
sufficiently fast, it can
still take much time to compute a (near-)optimal solution
to that problem (large linear programs
may have to be solved).

One of the first and well-known results on oblivious routing is due to
Borodin and Hopcroft~\cite{borodin1985routing} who showed that competitive oblivious
routing algorithms require randomization, as deterministic algorithms
come with high lower bounds: given an unweighted
network with $n$ nodes and maximum degree $\Delta$,
there exists a (permutation) routing instance such that the congestion
induced by a given deterministic oblivious routing scheme is 
at least $\Omega(\sqrt{n}/\Delta^{3/2})$.
This result was improved by Kaklamanis et al.~\cite{kaklamanis1991tight}
to a lower bound of $\Omega(\sqrt{n}/\Delta)$. 

For randomized algorithms Valiant and Brebner~\cite{VB81} showed how to obtain
a polylogarithmic competitive ratio for the hypercube by routing to random
intermediate destinations. R{\"a}cke~\cite{Rae02}
presented the first oblivious routing scheme with
a polylogarithmic competitive ratio of $O(\log^3{n})$
in general
networks.
The paper by R{\"a}cke
was also the first to propose designing oblivious routing schemes 
based on cut-based \emph{hierarchical decompositions}. 
However,
R{\"a}cke's result is non-constructive in the sense that only
an exponential time algorithm was given to construct the
hierarchy.
This approach 
has subsequently  been used to obtain approximate solutions
for a variety of cut-related problems that seem very hard
on general graphs but that are efficiently solvable on
trees, see e.g.~\cite{alon2006general,andreev2009simultaneous,
bansal2011min,chekuri2004all,engelberg2006cut,khandekar2014advantage,
konemann2006unified,racke2008optimal}.
Polynomial-time
algorithms for constructing the hierarchical decomposition
were given by Bienkowski et al.~\cite{bienkowski2003practical} 
and Harrelson
et al.~\cite{harrelson2003polynomial}. 
However, none of these results provide an (asymptotically) optimal
competitive ratio.

Azar et al.~\cite{azar2004optimal}
gave a polynomial time algorithm that for a given graph computes the
optimal oblivious routing via a linear programming approach, i.e., without
using a hierarchical decomposition.

An optimal approximation
guarantee of $O(\log{n})$
(which matches a known lower bound from grids~\cite{bartal1997line,maggs1997exploiting})
was first presented by  R{\"a}cke~\cite{racke2008optimal} .
Instead of
considering a single tree to approximate the cut-structure
of a graph $G$,
~\cite{racke2008optimal}  
proposes to use a convex combination
of decomposition trees. The
paper relies on multiplicative weight updates 
and the proof technique is similar to
the technique used by Charikar et al.~\cite{charikar1998approximating}
for finding a
probabilistic embedding of a metric into a small number of
dominating tree metrics. 

More recently, 
inspired by the ideas on cut matching games 
introduced by Khandekar, Rao,
and Vazirani~\cite{khandekar2009graph}, 
R{\"a}cke et al.~\cite{racke2014computing}
presented a \emph{fast} construction algorithm for hierarchical
tree decompositions, i.e., for a single tree: 
given an undirected graph $G = (V, E,c)$ 
with edge capacities, a single tree $T = (V_T, E_T, c_T)$ 
can be computed whose leaf nodes correspond to nodes in $G$
and which approximates the cut-structure of $G$ up to
a factor of $O(\log^4{n})$ (i.e., the faster runtime comes at the price
of a worse approximation guarantee).
In particular, the authors 
present almost linear-time cut-based hierarchical decompositions,
by establishing a connection between approximation of max flow
and oblivious routing.
This overcomes the major drawback of earlier algorithms such as~\cite{harrelson2003polynomial} and even~\cite{madry2010fast} which
required high running times for constructing the decomposition
tree (or the distribution over decomposition
trees).
The bound has been improved further by 
Peng in~\cite{peng2016approximate}.

Previous result on compact oblivious routing strategy focuses on routing
strategies that aim to minimize the path-length instead of the congestion.
(see e.g.~\cite{cowen2001compact,krioukov2007compact,van1995compact}.
The research 
community has derived many interesting results on compact shortest path routing
on special graphs, e.g., characterizing hypercubes, trees, scale-free networks,
and planar graphs~\cite{fraigniaud2001routing,frederickson1988designing,gavoille2001routing,krioukov2004compact,thorup2001compact}.
However, it is also known that it is 
impossible to implement
shortest path routing with routing tables whose size in all
network topologies grows slower than linear with the increase
of the network size~\cite{fraigniaud1995memory,gavoille1996memory}. 

As a resort, compact routing research studies algorithms to decrease routing
table sizes at the price of letting packets to be routed
along suboptimal paths. In this context, suboptimal means
that the forwarding paths are allowed to be longer than the
shortest ones, but the length increase is bounded by a constant
stretch factor. A particularly interesting result is by
Thorup et al.~\cite{thorup2001compact} who
presented compact routing schemes for general weighted
undirected networks, ensuring small routing tables, 
small headers and low stretch.
The approach relies on an interesting 
shortest path routing scheme for trees of arbitrary degree
and diameter that assigns each vertex of an $n$-node tree a
 label of logarithmic size. Given the label of a source node
and the label of a destination it is possible to compute, in
constant time, the port number of the edge from the source
that heads in the direction of the destination.
An interesting recent work by Retvari et al.~\cite{retvari2013compact}
generalizes compact routing to arbitrary routing policies that favor a broader
set of path attributes beyond path length. Using routing algebras,
the authors identify the algebraic requirements for
a routing policy to be realizable with sublinear size routing tables. 

\section{Conclusion}\label{sec:conclusion}

Given the fast growth of communication networks (e.g., due the 
advent of novel paradigms such as Internet-of-Things),
the high costs of network equipment
(e.g., fast memory is expensive and power hungry),
as well as 
the increasing miniaturization of communication-enabled devices,
we in this paper initiated the study of oblivious routing
schemes which only require small routing tables.
In particular, we presented the first \emph{compact},
oblivious routing scheme, requiring polylogarithmic
tables only (as well as polylogarithmic packet headers and vertex
labels).

We believe that our work opens an interesting avenue for future
research. In particular, while our algorithms provide
poly-logarithmic routing tables and competitive ratios,
it may be possible to further improve these results by logarithmic
factors. Furthermore, it would be interesting to generalize
our results to non-uniform network capacities, as well as to explore
whether our results can be improved for 
special network topologies arising in practice.

{\footnotesize
\bibliographystyle{plain}
\bibliography{compact,references}
}

\appendix
\section*{Appendix}
\begin{lemma}
\label{lem:chernoff}
Let $X_1,\dots,X_n$ denote a set of \emph{negatively correlated} random
variables taking values in the range $[0,1]$. Let $X$ denote their sum
and let $\mu \le E[X]$ denote a lower bound on the expectation of $X$.
Then for any $\delta\ge 1$ 
\begin{equation*}
\Pr[X\ge (1+\delta)\mu] \le e^{-\delta\mu/3}\enspace.
\end{equation*}
\end{lemma}

\end{document}